\documentclass[envcountsame]{llncs}
\usepackage{epsfig}
\usepackage{graphicx}
\usepackage{color}
\usepackage{url}
\usepackage{amsfonts}
\usepackage[utf8]{inputenc}
\usepackage{amsmath}
\usepackage{xspace}
\usepackage{amssymb}
\usepackage{amstext}
\usepackage{todonotes}
\usepackage{complexity}
\usepackage{enumerate}
\usepackage{stmaryrd}

\usepackage{tikz}
\usetikzlibrary{positioning}    

\def\ZVAS{$\mathbb{Z}$-VAS}
\def\ZVASS{$\zint$-\text{VASS}}
\def\ZVASSr{$\mathbb{Z}$-VASS$_{\text{R}}$}
\def\VASS{VASS}
\def\ZRM{$\zint$-RM}

\newcommand{\nat}{\mathbb{N}}
\newcommand{\zint}{\mathbb{Z}}

\newcommand{\dotcup}{\ensuremath{\uplus}}

\newcommand{\eval}[1]{\llbracket#1 \rrbracket}

\newcommand{\problemx}[3]{
\par\noindent\underline{\sc#1}\par\nobreak\vskip.2\baselineskip
\begingroup\clubpenalty10000\widowpenalty10000
\setbox0\hbox{\bf INPUT: }\setbox1\hbox{\bf QUESTION: }
\dimen0=\wd0\ifnum\wd1>\dimen0\dimen0=\wd1\fi
\vskip-\parskip\noindent
\hbox to\dimen0{\box0\hfil}\hangindent\dimen0\hangafter1\ignorespaces#2\par
\vskip-\parskip\noindent
\hbox to\dimen0{\box1\hfil}\hangindent\dimen0\hangafter1\ignorespaces#3\par
\endgroup}

\newcommand{\problemy}[3]{
\par\noindent\underline{\sc#1}\par\nobreak\vskip.2\baselineskip
\begingroup\clubpenalty10000\widowpenalty10000
\setbox0\hbox{\bf INPUT: }\setbox1\hbox{\bf OUTPUT: }
\dimen0=\wd0\ifnum\wd1>\dimen0\dimen0=\wd1\fi
\vskip-\parskip\noindent
\hbox to\dimen0{\box0\hfil}\hangindent\dimen0\hangafter1\ignorespaces#2\par
\vskip-\parskip\noindent
\hbox to\dimen0{\box1\hfil}\hangindent\dimen0\hangafter1\ignorespaces#3\par
\endgroup}

\begin{document}

\title{Integer Vector Addition Systems with States}
\author{Christoph Haase\thanks{Supported by the
    French ANR, \textsc{ReacHard}
    (grant ANR-11-BS02-001).} \and Simon Halfon}
\institute{Laboratoire Sp\'ecification et V\'erification (LSV), CNRS\\ 
  \'Ecole Normale Sup\'erieure (ENS) de Cachan, France}

\maketitle

\begin{abstract}
  This paper studies reachability, coverability and inclusion problems
  for Integer Vector Addition Systems with States (\ZVASS) and
  extensions and restrictions thereof. A \ZVASS \ comprises a
  finite-state controller with a finite number of counters ranging
  over the integers. Although it is folklore that reachability in
  \ZVASS \ is \NP-complete, it turns out that despite their
  naturalness, from a complexity point of view this class has received
  little attention in the literature. We fill this gap by providing an
  in-depth analysis of the computational complexity of the
  aforementioned decision problems. Most interestingly, it turns out
  that while the addition of reset operations to ordinary VASS leads
  to undecidability and Ackermann-hardness of reachability and
  coverability, respectively, they can be added to \ZVASS \ while
  retaining \NP-completeness of both coverability and reachability.
\end{abstract}


\section{Introduction}

Vector Addition Systems with States (\VASS) are a prominent class of
infinite-state systems. They comprise a finite-state controller with a
finite number of counters ranging over the natural numbers. When
taking a transition, an integer can be added to a counter, provided
that the resulting counter value is non-negative. A configuration of a
\VASS \ is a tuple $q(\vec{v})$ consisting of a control state $q$ and
a vector $\vec{v}\in \nat^d$, where $d>0$ is the number of counters
or, equivalently, the dimension of the VASS. The central decision
problems for \VASS \ are reachability, coverability and
inclusion. Given configurations $q(\vec{v})$, $q'(\vec{v}')$ of a
\VASS \ $\mathcal{A}$, reachability is to decide whether there is a
path connecting the two configurations in the transition system
induced by $\mathcal{A}$. Coverability on the other hand asks whether
there is a path from $q(\vec{v})$ to a configuration that is ``above''
$q'(\vec{v}')$, \emph{i.e.}, a path to some $q'(\vec{w})$ such that
$\vec{w}\ge \vec{v}'$, where $\ge$ is interpreted component-wise.
Finally, given \VASS \ $\mathcal{A}$ and $\mathcal{B}$, inclusion asks
whether the set of counter values reachable in the transition system
induced by $\mathcal{A}$ is contained in those reachable by
$\mathcal{B}$. All of the aforementioned problems have extensively
been studied over the course of the last forty years. One of the
earliest results was obtained by Lipton, who showed that reachability
and coverability are \EXPSPACE-hard~\cite{Lipt76}. Later, Rackoff
established a matching upper bound for coverability~\cite{Rack78}, and
Mayr showed that reachability is decidable~\cite{Mayr81}. For
inclusion, it is known that this problem is in general
undecidable~\cite{Hack76} and Ackermann
($\mathbf{F}_{\omega}$)-complete~\cite{Jan01} when restricting to
\VASS \ with a finite reachability set. Moreover, various extensions
of VASS with, for instance, resets or polynomial updates on counter
values have been studied in the literature. Resets allow for setting a
counter to zero along a transition, and polynomial updates allow for
updating a counter with an arbitrary polynomial. In general,
reachability in the presence of any such extension becomes
undecidable~\cite{DFS98,FGH13}, while the complexity of coverability
increases significantly to $\mathbf{F}_{\omega}$-completeness in the
presence of resets~\cite{Schn10}.

What makes \VASS \ hard to deal with, both in the computational and in
the mathematical sense, is the restriction of the counters to
non-negative integers. This restriction allows for enforcing an order
in which transitions can be taken, which is at the heart of many
hardness proofs.  In this paper, we relax this restriction and study
\ZVASS \ which are VASS whose counters can take values from the
integers, and extensions thereof. Thus, the effect of transitions can
commute along a run of a \ZVASS, which makes deciding reachability
substantially easier, and it is in fact folklore that reachability in
\ZVASS \ is \NP-complete.  It appears, however, that many aspects of
the computational complexity of standard decision problems for \ZVASS
\ and extensions and restrictions thereof have not received much
attention in the literature.

\subsubsection*{Our contribution.} The main focus of this paper\footnote{A
full version containing all proofs omitted due to space constraints
can be obtained from \url{http://arxiv.org/abs/1406.2590}.}
is to study the computational complexity of reachability, coverability
and inclusion for \ZVASS \ equipped with resets (\ZVASSr). Unlike in
the case of \VASS, we can show that reachability and coverability are
naturally logarithmic-space inter-reducible. By generalizing a
technique introduced by Seidl~\emph{et al.}~\cite{SSMH04} for defining
Parikh images of finite-state automata in existential Presburger
arithmetic, we can show that a given instance of reachability (and
\emph{a fortiori} coverability) in \ZVASSr \ can be reduced in
logarithmic-space to an equivalent sentence in existential Presburger
arithmetic, and henceforth both problems are $\NP$-complete. Moreover,
by exploiting a recent result on the complexity of Presburger
arithmetic with a fixed number of quantifier
alternations~\cite{Haa14}, this reduction immediately yields
$\coNEXP$-membership of the inclusion problem for \ZVASSr. We also
show that a matching lower bound can be established via a reduction
from validity in $\Pi_2$-Presburger arithmetic. This lower bound does
not require resets and thus already holds for \ZVASS. Along the way,
wherever possible we sharpen known lower bounds and propose some
further open problems.

\subsubsection*{Related Work.} The results obtained in this paper are
closely related to decision problems for commutative grammars,
\emph{i.e.}  Parikh images of, for instance, finite-state automata or
context-free grammars. A generic tool that is quite powerful in this
setting is to define Parikh images as the set of solutions to certain
systems of linear Diophantine equations. This approach has, for
instance, been taken in~\cite{Esp97,PR99,SSMH04,HKOW09,HL11}. As
stated above, we generalize the technique of Seidl~\emph{et al.},
which has also been the starting point in~\cite{HL11} in order to show
decidability and complexity results for pushdown systems equipped with
reversal-bounded counters.

Furthermore, results related to ours have also been established by
Kopczy{\'n}ski \& To. In~\cite{KT10}, they consider inclusion problems
for regular and context-free commutative grammars, and show that for a
fixed alphabet those problems are $\coNP$- and
$\mathsf{\Pi}_2^\P$-complete, respectively. As a matter of fact, the
proof of the $\mathsf{\Pi}_2^\P$-upper bound is established for
context-free commutative grammars in which, informally speaking,
letters can be erased, which can be seen as a generalization of
\ZVASS. In general, inclusion for context-free commutative grammars is
in \coNEXP~\cite{Huy85}, but it is not known whether this bound is
tight. Also related is the work by Reichert~\cite{Rei13}, who studies
the computational complexity of reachability games on various classes
of \ZVASS. Finally, \ZVASS \ are an instance of valence automata,
which have recently, for instance, been studied by Buckheister \&
Zetzsche~\cite{BZ13}. However, their work is more concerned with
language-theoretic properties of valence automata rather than aspects
of computational complexity. Language-theoretic aspects of \ZVASS
\ have also been studied by Greibach~\cite{Grei78}.

As discussed above, \ZVASS \ achieve a lower complexity for standard
decision problems in comparison to \VASS \ by relaxing counters to
range over the integers. Another approach going into a similar
direction is to allow counters to range over the positive reals. It
has been shown in recent work by Fraca \& Haddad~\cite{FH13} that the
decision problems we consider in this paper become substantially
easier for such continuous \VASS, with reachability even being
decidable in \P.



\section{Preliminaries}
\label{sec:def}
In this section, we provide most of the definitions that we rely on in
this paper. We first introduce some general notation and subsequently
an abstract model of register machines from which we derive \ZVASS \ as a
special subclass. We then recall and tighten some known complexity
bounds for \ZVASS \ and conclude this section with a brief account on
Presburger arithmetic.

\subsubsection*{General Notation.} 
In the following, $\zint$ and $\nat$ are the sets of integers and
natural numbers, respectively, and $\nat^d$ and $\zint^d$ are the set
of dimension $d$ vectors in $\nat$ and $\zint$, respectively. We
denote by $[d]$ the set of positive integers up to $d$, \emph{i.e.}
$[d]=\{1,\dots,d\}$. By $\nat^{d \times d}$ and $\zint^{d\times d}$ we
denote the set of $d\times d$ square matrices over $\nat$ and $\zint$,
respectively. The identity matrix in dimension $d$ is denoted by $I_d$
and $\vec{e}_i$ denotes the $i$-th unit vector in any dimension $d$
provided $i\in [d]$. For any $d$ and $i,j\in [d]$, $E_{ij}$ denotes
the $d\times d$-matrix whose $i$-th row and $j$-th column intersection is equal to
one and all of its other components are zero, and we use $E_i$ to
abbreviate $E_{ii}$. For $\vec{v} \in \zint^d$ we write $\vec{v}(i)$
for the $i$-th component of $\vec{v}$ for $i\in [d]$. Given two
vectors $\vec{v_1}, \vec{v_2} \in \zint^d$, we write $\vec{v_1} \ge
\vec{v_2}$ iff for all $i\in [d]$, $\vec{v_1}(i) \ge \vec{v_2}(i)$.
Given a vector $\vec{v} \in \zint^d$ and a set $S \subseteq [d]$, by
$\vec{v}_{|S}$ we denote the vector $\vec{w}$ derived from $\vec{v}$
with components from $S$ reset, \emph{i.e}, for all $j \in [d]$,
$\vec{w}(j) = \vec{v}(j)$ when $j \notin S$, and $\vec{w}(j) = 0$
otherwise. Given $i \in [d]$, $\vec{v}_{|i}$ abbreviates
$\vec{v}_{|\{i\}}$. If not stated otherwise, all numbers in this paper
are assumed to be encoded in binary.

\subsubsection*{Presburger Arithmetic.} 
Recall that \emph{Presburger arithmetic (PA)} is the first-order
theory of the structure $\langle \nat, 0, 1, +, \ge\rangle$,
\emph{i.e.}, quantified linear arithmetic over natural numbers. The
size $|\Phi|$ of a PA formula is the number of symbols required to
write it down, where we assume unary encoding of numbers\footnote{This
  is with no loss of generality since binary encoding can be simulated
  at the cost of a logarithmic blowup of the formula size. Note that
  in particular all complexity lower bounds given in this paper still hold
  assuming unary encoding of numbers.}.  For technical convenience, we
may assume with no loss of generality that terms of PA formulas are of
the form $\vec{z}\cdot \vec{x}\ge b$, where $\vec{x}$ is an $n$-tuple
of first-order variables, $\vec{z}\in \zint^n$ and $b\in \zint$. It is
well-known that the existential ($\Sigma_1$-)fragment of PA is
\NP-complete, see \emph{e.g.}~\cite{BT76}. Moreover, validity for the
$\Pi_2$-fragment of PA, \emph{i.e.} its restriction to a
$\forall^*\exists^*$-quantifier prefix, is
$\coNEXP$-complete~\cite{Grae89,Haa14}.

Given a PA formula $\Phi(x_1,\ldots,x_d)$ in $d$ free variables, we
define
\begin{align*}
  \eval{\Phi(x_1,\ldots,x_d)} & = \{ (n_1,\ldots,n_d)\in \nat^d : 
  \Phi(n_1/x_1,\ldots,n_d/x_d) 
  \text{ is valid}\}.
\end{align*}
Moreover, a set $M\subseteq \nat^d$ is \emph{PA-definable} if there
exists a PA formula $\Phi(x_1,\ldots,x_d)$ such that
$M=\eval{\Phi(x_1,\ldots,x_d)}$. Recall that a result due to Ginsburg
\& Spanier states that PA-definable sets coincide with the so-called
\emph{semi-linear sets}~\cite{GS66}.

\subsubsection*{Integer Vector Addition Systems.} 

The main objects studied in this paper can be derived from a general
class of integer register machines which we define below.
\begin{definition} 
  Let $\mathfrak{A}\subseteq \zint^{d\times d}$, a \emph{dimension
    $d$-integer register machine over $\mathfrak{A}$
    (\ZRM$(\mathfrak{A})$)} is a tuple $\mathcal{A}=(Q, \Sigma, d,
  \Delta, \tau)$ where
  \begin{itemize}
  \item $Q$ is a finite set of \emph{control states},
  \item $\Sigma$ is a finite \emph{alphabet},
  \item $d>0$ is the \emph{dimension} or the number of \emph{counters},
  \item $\Delta \subseteq Q \times \Sigma \times Q$ is a finite set of
    \emph{transitions},
  \item $\tau: \Sigma \to (\zint^d \to \zint^d)$ maps each $a\in
    \Sigma$ to an \emph{affine transformation} such that $\tau(a) =
    \vec{v} \mapsto A\vec{v} + \vec{b}$ for some $A\in \mathfrak{A}$
    and $\vec{b}\in \zint^d$.
  \end{itemize}
\end{definition}

We will often consider $\tau$ as a morphism from $\Sigma^*$ to the set
of affine transformations such that $\tau(\epsilon)=I_d$ and for any
$w\in \Sigma^*$ and $a\in \Sigma$, $\tau(wa)(\vec{v}) =
\tau(a)(\tau(w)(\vec{v}))$. The set $C(\mathcal{A})=Q\times \zint^d$
is called the \emph{set of configurations of $\mathcal{A}$}. For
readability, we write configurations as $q(\vec{v})$ instead of
$(q,\vec{v})$. Given configurations $q(\vec{v}), q'(\vec{v}')\in C$,
we write $q(\vec{v})\stackrel{a}{\rightarrow}_\mathcal{A}q(\vec{v}')$
if there is a transition $(q,a,q')\in \Delta$ such that $\vec{v}' =
\tau(a)(\vec{v})$, and $q(\vec{v}) \rightarrow_{\mathcal{A}}
q'(\vec{v}')$ if
$q(\vec{v})\stackrel{a}{\rightarrow}_\mathcal{A}q(\vec{v}')$ for some
$a\in \Sigma$. A \emph{run on a word $\gamma=a_1\cdots a_n\in
  \Sigma^*$} is a finite sequence of configurations $\varrho: c_0 c_1
\cdots c_n$ such that
$c_i\stackrel{a_{i+1}}{\rightarrow}_{\mathcal{A}} c_{i+1}$ for all
$0\le i<n$, and we write $c_0
\stackrel{\gamma}{\rightarrow}_\mathcal{A} c_n$ in this
case. Moreover, we write $c \rightarrow^*_\mathcal{A} c'$ if there is
a run $\varrho$ on some word $\gamma$ such that $c=c_0$ and
$c'=c_n$. Given $q(\vec{v})\in C(\mathcal{A})$, the \emph{reachability
  set starting from $q(\vec{v})$} is defined as
\begin{align*}
  \mathit{reach}(\mathcal{A},q(\vec{v})) = \{ \vec{v}'\in \zint^d :
  q(\vec{v}) \rightarrow^*_{\mathcal{A}} q'(\vec{v}') \text{ for some
  } q'\in Q\}.
\end{align*}



In this paper, we study the complexity of deciding reachability,
coverability and inclusion.
\vspace*{0.25cm}
\problemx{\ZRM$(\mathfrak{A})$ Reachability/Coverability/Inclusion}
  {\ZRM$(\mathfrak{A})$ $\mathcal{A}$, $\mathcal{B}$, configurations 
    $q(\vec{v}),q'(\vec{v}')\in C(\mathcal{A})$, $p(\vec{w})\in C(\mathcal{B})$.}
  {\emph{Reachability:} Is there a run $q(\vec{v})\rightarrow^*_\mathcal{A} q'
    (\vec{v}')$?\\
    \emph{Coverability:} Is there a $\vec{z}\in \zint^d$ s.t.
    $q(\vec{v}) \rightarrow^*_\mathcal{A} q'(\vec{z})$ and $\vec{z}\ge \vec{v}'$?\\
    \emph{Inclusion:} Does $\mathit{reach}(\mathcal{A},q(\vec{v})) \subseteq 
    \mathit{reach}(\mathcal{B}, p(\vec{w}))$ hold? 
  }
\vspace*{0.25cm}
If we allow an arbitrary number of control states, whenever it is
convenient we may assume $\vec{v},\vec{v}'$ and $\vec{w}$ in the
definition above to be equal to $\vec{0}$. Of course, \ZRM \ are very
general, and all of the aforementioned decision problems are already
known to be undecidable, we will further elaborate on this topic
below. We therefore consider subclasses of \ZRM$(\mathfrak{A})$ in
this paper which restrict the transformation mappings or the number of
control states: $\mathcal{A}$ is called
\begin{itemize}
\item \emph{integer vector addition system with states and resets
  (\ZVASSr)} if $\mathfrak{A} =  \{ \lambda_1E_1 + \cdots
  + \lambda_dE_d : \lambda_i\in \{0,1\}, i\in [d]\}$;
\item \emph{integer vector addition system with states (\ZVASS)} if
  $\mathfrak{A} = I_d$;
\item \emph{integer vector addition system (\ZVAS)} if $\mathcal{A}$
  is a \ZVASS \ and $|Q|=1$.
\end{itemize}
Classical vector addition systems with states (\VASS) can be recovered
from the definition of \ZVASS \ by defining the set of configurations
as $Q\times \nat^d$ and adjusting the definition of
$\rightarrow_{\mathcal{A}}$ appropriately. It is folklore that
coverability in \VASS \ is logarithmic-space reducible to reachability
in \VASS. Our first observation is that unlike in the case of VASS,
reachability can be reduced to coverability in \ZVASS, this even holds
for \ZVASSr.
Thanks to this observation, all lower and upper bounds for
reachability carry over to coverability, and \emph{vice versa}.
\begin{lemma}
\label{equiv}
  Reachability and coverability are logarithmic-space inter-reducible
  in each of the classes \ZVASSr, \ZVASS \ and \ZVAS. The reduction
  doubles the dimension.
\end{lemma}
\begin{proof}
  The standard folklore construction to reduce coverability in \VASS
  \ to reachability in \VASS \ also works for all classes of
  \ZVASSr. For brevity, we therefore only give the reduction in the
  converse direction.

  Let $\mathcal{A}$ be from any class of \ZVASS \ in dimension $d$ and
  let $q(\vec{v}),q'(\vec{v}')\in C(\mathcal{A})$. We construct a
  \ZVASS \ $\mathcal{B}$ in dimension $2d$ with the property
  $q(\vec{v})\rightarrow^*_{\mathcal{A}} q'(\vec{v}')$ iff
  $q(\vec{v},-\vec{v})\rightarrow^*_{\mathcal{B}}
  q'(\vec{v}',-\vec{v}')$ as follows: any affine transformation
  $\vec{v} \mapsto A\vec{v} + \vec{b}$ is replaced by $\vec{v} \mapsto
  A'\vec{v} + \vec{b}'$, where
  \begin{align*}
    A' & = \begin{bmatrix}
      A & \vec{0}\\
      \vec{0} & A
    \end{bmatrix}
    & \vec{b}' = \begin{vmatrix} \vec{b} \\ -\vec{b}\end{vmatrix}.
  \end{align*}
  
  Any run $\varrho: q_0(\vec{v}_0)\cdots q_n(\vec{v}_n)$ in
  $\mathcal{B}$ such that $q_0(\vec{v}_0)=q(\vec{v},-\vec{v})$ and
  $q_n(\vec{v}_n)=q'(\vec{v}',-\vec{v}')$ corresponds in the first $d$
  components to a run in $\mathcal{A}$. Moreover, $\varrho$ has the
  property that for any $0\le i\le n$ and $q_i(\vec{v}_i)$,
  $\vec{v}_i(j)=-\vec{v}_i(j+d)$ for all $j\in [d]$. Therefore,
  $q(\vec{v},-\vec{v}) \rightarrow_\mathcal{B}^* q'(\vec{w},-\vec{w})$
  for some $q'(\vec{w},-\vec{w})$ that covers $q'(\vec{v}',-\vec{v}')$
  if, and only if, $\vec{w} \ge \vec{v}'$ and $-\vec{w}\ge -\vec{v}'$,
  \emph{i.e.}, $\vec{w}=\vec{v}'$ and thus in particular whenever
  $\mathcal{A}$ reaches $q'(\vec{v}')$ from $q(\vec{v})$.\qed
\end{proof}

\subsubsection*{Known Complexity Results for \ZVASS.} 
It is folklore that reachability in \ZVASS \ is \NP-hard. Most
commonly, this is shown via a reduction from \textsc{Subset Sum}, so
this hardness result in particular relies on binary encoding of
numbers and the presence of control states. Here, we wish to remark
the following observation.
\begin{lemma}
  \label{lem:np-hardness}
  Reachability in \ZVAS\ is \NP-hard even when numbers are encoded in
  unary.
\end{lemma}
The proof is given in the appendix of the full version of this paper
and follows straight-forwardly via a reduction from feasibility of a
system of linear Diophantine equations $A\vec{x}=\vec{b}, \vec{x}\ge
\vec{0}$, which is known to be \NP-complete even when unary encoding
of numbers is assumed~\cite{GJ79}. Apart from that, it is folklore
that reachability in \ZVASS \ is in \NP. To the best of the authors'
knowledge, no upper bounds for reachability, coverability or inclusion
for \ZVASSr \ have been established so far.

Next, we recall that slightly more general transformation matrices
lead to undecidability of reachability: when allowing for arbitrary
diagonal matrices, \emph{i.e.} affine transformations along
transitions, reachability becomes undecidable already in dimension
two~\cite{FGH13}. Consequently, by a straight forward adaption of
Lemma~\ref{equiv} we obtain the following.

\begin{lemma}
  \label{lem:undecidability}
  Let $\mathfrak{D}_d$ be the set of all diagonal matrices in
  dimension $d$. Coverability in \ZRM$(\mathfrak{D}_d)$ is undecidable
  already for $d=4$.
\end{lemma}
Of course, undecidability results for reachability in matrix
semi-groups obtained in~\cite{BP08} can be applied in order to obtain
undecidability results for more general classes of matrices, and those
undecidability results do not even require the presence of control
states.



\section{Reachability in \ZVASSr\ is in \NP}
\label{sec:zvass}

One main idea for showing that reachability for \ZVASSr \ is in \NP
\ is that since there are no constraints on the counter values along a
run, a reset on a particular counter allows to forget any information
about the value of this counter up to this point, \emph{i.e.}, a reset
cuts the run. Hence, in order to determine the value of a particular
counter at the end of a run, we only need to sum up the effect of the
operations on this counter since the last occurrence of a reset on
this counter. This in turn requires us to guess and remember the last
occurrence of a reset on each counter.

Subsequently, we introduce monitored alphabets and generalized Parikh
images in order to formalize our intuition behind resets. A
\emph{monitored alphabet} is an alphabet $\Sigma \dotcup R$ with $R =
\{r_1, \dots, r_k \}$ being the monitored letters. Given $S \subseteq
   [k]$, we denote by $\Sigma_S = \Sigma \cup \{ r_i : i\in S\}$ the
   alphabet containing only monitored letters indexed from $S$. Any
   word $\gamma \in (\Sigma \cup R)^*$ over a monitored alphabet
   admits a unique decomposition into \emph{partial words}
\begin{align*}
  \gamma = \gamma_0 r_{i_1} \gamma_1 r_{i_2} \cdots r_{i_\ell} \gamma_\ell
\end{align*}
for some $\ell\le k$ such that all $i_j$ are pairwise distinct and for
all $j \in [\ell]$, $\gamma_j \in \Sigma_{\{ r_{i_{j+1}}, \dots,
  r_{i_\ell}\}}^* $. Such a decomposition simply keeps track of the
last occurrence of each monitored letter. For instance for $k=4$ and
$\Sigma = \{a,b\}$, the word $\gamma = aabr_1br_3abr_3ar_1$ can
uniquely be decomposed as $(aabr_1br_3ab)r_3(a)r_1$.

In this paper, the \emph{Parikh image} $\pi_\Sigma(w)$ of a word $w\in
(\Sigma\dotcup R)^*$ restricted to the alphabet $\Sigma = \{a_1,
\dots, a_n\}$ is the vector $\pi_\Sigma(w) \in \nat^n$ such that
$\pi(w)(i) = |w|_{a_i}$ is the number of occurrences of $a_i$ in
$w$. Moreover, $\mathfrak{S}_k$ denotes the permutation group on $k$
symbols.


\begin{definition}
\label{def}
Let $\Sigma\dotcup R$ be a monitored alphabet such that $|\Sigma|=n$
and $|R|=k$. A tuple $(\vec{\alpha},\sigma) =
(\vec{\alpha}_0,\vec{\alpha}_1,\dots, \vec{\alpha}_k, \sigma) \in
(\nat^n)^{k+1} \times \mathfrak{S}_k$ is \emph{a generalized Parikh
  image of $\gamma \in(\Sigma \dotcup R)^*$} if there exist $0 \le p
\le k$ and a decomposition $\gamma = \gamma_p r_{\sigma(p+1)}
\gamma_{p+1}r_{\sigma(p+2)} \cdots r_{\sigma(k)} \gamma_k$ such that:
\begin{enumerate}[(a)]
\item for all $p\le i \le k$, $\gamma_i \in \Sigma_{R_i}^*$, where
  $R_i=\{r_{\sigma(i+1)}, \dots, r_{\sigma(k)}\}$; and
\item for all $0\le i < p$, $\vec{\alpha}_i = \vec{0}$ and for all $p
  \le i \le k$, $\vec{\alpha}_i = \pi_\Sigma(\gamma_i)$, the Parikh
  image of $\gamma_i$ restricted to $\Sigma$, \emph{i.e.} monitored
  alphabet symbols are ignored.
\end{enumerate}
The generalized Parikh image of a language $L\subseteq (\Sigma\dotcup
R)^*$ is the set $\Pi(L) \subseteq (\nat^n)^{k+1} \times
\mathfrak{S}_k$ of all generalized Parikh images of all words
$\gamma\in L$.
\end{definition}
This definition formalizes the intuition given by the decomposition
described above with some additional padding of dummy vectors for
monitored letters not occurring in $\gamma$ in order to obtain
canonical objects of \emph{uniform size}. Even though generalized
Parikh images are not unique, two generalized Parikh images of the
same word only differ in the order of dummy monitored letters. For
instance for $k=4$, the word $\gamma = aabr_1br_3abr_3ar_1$ has two
generalized Parikh images: they coincide on $\vec{\alpha}_0 =
\vec{\alpha}_1 = \vec{\alpha}_2 = (0,0)$, $\vec{\alpha}_3 = (3,3)$,
$\vec{\alpha}_4 = (1,0)$ and $\sigma(3) = 3$, $\sigma(4) = 1$, and
only differ on $\sigma(1)$ and $\sigma(2)$ that can be $2$ and $4$, or
$4$ and $2$, respectively.

Generalized Parikh images can now be applied to reachability in
\ZVASSr \ as follows. Without loss of generality, we may assume that a
\ZVASSr \ in dimension $d$ is given as $\mathcal{A} = (Q, \Sigma
\dotcup R, d, \Delta, \tau)$ for some alphabet $\Sigma = \{a_1, \dots,
a_n\}$ and $R=\{r_1, \dots, r_d\}$ such that $\tau(r_i) = \vec{v}
\mapsto v_{|i}$ for any $i\in [d]$ and for any $a_i\in \Sigma$,
$\tau(a_i)=\vec{v} \mapsto \vec{v} + \vec{b}_i$ for some $\vec{b}_i
\in \zint^d$. This assumption allows for isolating transitions
performing a reset and enables us to apply monitored alphabets by
monitoring when a reset occurs in each dimension the last
time. Consequently, the counter value realized by some $\gamma\in
(\Sigma \dotcup R)^*$ starting from $\vec{0}$ is fully determined by a
generalized Parikh image of $\gamma$.
%
\begin{lemma}
  \label{lem:parikh-to-effect}
  Let $\mathcal{A}$ be a \ZVASSr, $\gamma \in (\Sigma\dotcup R)^*$,
  $(\vec{\alpha}_0,\vec{\alpha}_1, \dots, \vec{\alpha}_d, \sigma)\in
  \Pi(\gamma)$ and $B\in \zint^{d\times n}$ the matrix whose columns
  are the vectors $\vec{b}_i$. Then the following holds:
  \begin{align*}
    \tau(\gamma)(\vec{0}) = \sum\nolimits_{1\le i\le d}
    (B\vec{\alpha}_{i-1})_{|\{\sigma(i), \dots, \sigma(d)\}} +
    B\vec{\alpha}_d.
  \end{align*}
\end{lemma}

It thus remains to find a suitable way to define the generalized
Parikh image of the language of the non-deterministic finite state
automaton (NFA) underlying a \ZVASSr. In~\cite{SSMH04}, it is shown
how to construct in linear time an existential Presburger formula
representing the Parikh image of the language of an NFA. We generalize
this construction to generalized Parikh images of NFA over a monitored
alphabet, the original result being recovered in the absence of
monitored alphabet symbols, \emph{i.e.} when $k=0$. To this end, we
introduce below some definitions and two lemmas from the construction
provided in~\cite{SSMH04} which we employ for our
generalization. First, a \emph{flow} in an NFA $\mathcal{B} = (Q,
\Sigma, \Delta, q_0, F)$ is a triple $(f, s, t)$ where $s, t \in Q$
are states, and $f: \Delta \rightarrow \nat$ maps transitions $(p, a,
q) \in \Delta$ to natural numbers. Let us introduce the following
abbreviations:
\begin{align*}
\text{in}_f(q) = \sum_{(p,a,q)\in\Delta} f(p,a,q) \text{\quad and \quad} \text{out}_f(p) = \sum_{(p,a,q)\in\Delta} f(p,a,q).
\end{align*} 
A flow $(f,s,t)$ is called \emph{consistent} if for every $p\in Q$,
$\text{in}_f(p) = \text{out}_f(p) + h(p)$, where $h(p)=0$ for every
$p\in Q\setminus\{s,t\}$, and $h(s)=h(t)=0$ if $s=t$, and $h(s) = -1$
and $h(t) = 1$ otherwise. A flow is \emph{connected} if the undirected
graph obtained from the graph underlying the automaton when removing
edges with zero flow is connected. A consistent and connected flow
simply enforces Eulerian path conditions on the directed graph
underlying $\mathcal{B}$ so that any path starting in $s$ and ending
in $t$ yields a unique such flow.

\begin{lemma}[\cite{SSMH04}]
\label{lem1}
A vector $\vec{\alpha} \in \nat^n$ is in the Parikh image of
$\mathcal{L}(\mathcal{B})$ if, and only if, there is a consistent and
connected flow $(f,s,t)$ such that
\begin{itemize}
\item $s = q_0$, $t \in F$, and
\item for each $a_i\in \Sigma$, $\vec{\alpha}(i) = \sum_{(p,a_i,q)\in \Delta} f(p,a_i,q)$
\end{itemize}
\end{lemma}

Subsequently, in order to conveniently deal with states and alphabet
symbols in Presburger arithmetic, we write $Q = \{ \tilde{1}, \dots,
\tilde{m} \}$, $\Sigma = \{\dot{1}, \dots, \dot{n}\}$ and $R =
\{\dot{(n+1)},\dots,\dot{(n+k)}\}$. This enables us to write within
the logic terms like $p = q$ for $\tilde{p}, \tilde{q} \in
Q$. Moreover, it is easy to construct a formula $\varphi_\Delta(p, a,
q)$ such that $\varphi_\Delta(p,a,q)$ holds if, and only if,
$(\tilde{p}, \dot{a}, \tilde{q}) \in \Delta$. In particular,
$\varphi_\Delta$ can be constructed in linear time, independently of
the encoding of the NFA and its graph structure. With this encoding,
it is not difficult to see how the conditions from Lemma \ref{lem1}
can be checked by an existential Presburger formula.
\begin{lemma}[\cite{SSMH04}]
\label{lem2}
There exists a linear-time computable existential Presburger formula
$\varphi_\mathcal{B}(\vec{f},s,t)$ with the following properties:
\begin{itemize}
\item $\vec{f}$ represents a flow, \emph{i.e.}, is a tuple of
  variables $x_{(p,a,q)}$ for each $(p,a,q)\in\Delta$;
\item $s$ and $t$ are free variables constrained to represent states
  of $Q$; and
\item $(m_{\delta_1}, \dots, m_{\delta_g}, m_s, m_t)\in
  \eval{\varphi_\mathcal{B}(\vec{f},s,t)}$ if, and only if, the flow
  $(f_m,\tilde{m_s},\tilde{m_t})$ defined by $f_m(\delta_i) = m_{\delta_i}$ is 
  connected and consistent in $\mathcal{B}$.
\end{itemize}
\end{lemma}
We can now show how to generalize the construction from~\cite{SSMH04}
to monitored alphabets and generalized Parikh images.  Subsequently,
recall that $k$ is the number of monitored letters.

\begin{theorem}
\label{gpi}
Given an NFA $\mathcal{B} = (Q, \Sigma \dotcup R, \Delta, \tilde{q_0},
F)$ over a monitored alphabet $\Sigma \dotcup R$, an existential
Presburger formula $\Psi_\mathcal{B}(\vec{\alpha}, \vec{\sigma})$
defining the generalized Parikh image of the language
${\mathcal{L}(B)}$ of $\mathcal{B}$ can be constructed in time
$O(k^2|\mathcal{B}|)$.
\end{theorem}

\begin{proof}
  The formula we construct has free variables $\alpha_0^1, \dots,
  \alpha_0^n, \alpha_1^1, \dots, \alpha_k^n$ representing the $k+1$
  vectors $\vec{\alpha_0}, \dots, \vec{\alpha_k}$ and free variables
  $\vec{\sigma}=(\sigma_1, \dots, \sigma_k)$ to represent the
  permutation $\sigma$. First, we construct a formula
  $\varphi_{\text{perm}}$ expressing that $\vec{\sigma}$ is a
  permutation from $[k]$ to $[k]$:
\begin{align*}
  \varphi_{\text{perm}}(\vec{\sigma}) = \bigwedge_{i\in [k]}\left( 1 \le \sigma_i \le k \ \wedge
  \bigwedge\nolimits_{j \in [k]} i \neq j \rightarrow \sigma_i \neq \sigma_j \right).
\end{align*}
This formula has already size $O(k^2)$. Now we have to compute the
flow for each of the $k+1$ parts of the runs corresponding to the
$k+1$ partial words, but first we have to ``guess'' the starting and
ending states of each of these partial runs, in order to use the
formula from Lemma~\ref{lem2}. Let $\vec{s}=(s_0,\ldots,s_k)$ and
$\vec{t}=(t_0,\ldots,t_k)$, we define
\begin{multline*}
  \varphi_{\text{states}}(\vec{\sigma}, p, \vec{s}, \vec{t}) = s_0 = q_0 \wedge 
  \bigvee\nolimits_{\tilde{q}\in F} t_k = q \wedge
  \\ \bigwedge_{i \in [k]} [i \le p \rightarrow s_{i-1} = t_{i-1} \wedge t_{i-1}
  = s_i] \wedge [p < i \rightarrow \varphi_\Delta(t_{i-1},
  n + \sigma_i, s_i)].
\end{multline*}
Here, $p$ is used as in Definition~\ref{def}. We can now express the
$k+1$ flows: we need one variable per transition for each partial run.
\begin{multline*}
  \varphi_\text{flows}(\vec{\sigma}, p, \vec{f}, \vec{s}, \vec{t}) = 
  \bigwedge_{0\le i\le k} i < p
  \rightarrow \sum_{(p,a,q) \in \Delta} x_{(p,a,q)}^i = 0 \wedge \\
  \wedge \bigwedge_{0\le i\le
    k} p \le i \rightarrow \left( \varphi_{\mathcal{B}}(\vec{f}_i,s_i,
  t_i,) \wedge
  \bigwedge_{1\le j < i} \bigwedge_{(p,\dot{a}, q) \in\Delta} a = n + \sigma_j 
  \rightarrow x_{(p,\dot{a},q)}^i = 0\right),
\end{multline*}
where $\vec{f} = (\vec{f}_0, \dots, \vec{f}_k)$ and $\vec{f}_i$ is the
tuple of free variables of the form $x_{(p,a,q)}^i$ for all
$(p,a,q)\in\Delta$. This formula essentially enforces the constraints
from Definition~\ref{def}. The first line enforces that the ``dummy
flows'' $\vec{f}_0,\ldots,\vec{f}_{p-1}$ have zero flow. The second
line ensures that the flows $\vec{f}_p,\ldots,\vec{f}_k$ actually
correspond to partial words $\gamma_i$ in the decomposition described
in Definition~\ref{def}, and that monitored letters that, informally
speaking, have expired receive zero flow. Now putting everything
together yields:
\begin{multline*}
  \Psi_\mathcal{B}(\vec{\alpha}, \vec{\sigma}) = \exists p,\vec{f}_0,\ldots
  \vec{f}_k,\vec{s},\vec{t}.\,
  0 \le p \le k \wedge
  \varphi_\text{perm}(\vec{\sigma}) \wedge\\ \wedge
  \varphi_\text{states}(\vec{\sigma}, p, \vec{s}, \vec{t}) \wedge
  \varphi_\text{flows}(\vec{\sigma}, p, \vec{f}, \vec{s}, \vec{t}) 
  \wedge \bigwedge_{0
    \le i \le k} \bigwedge_{a \in [n]} \alpha_i^a = \sum_{(p,\dot{a},q)\in\Delta}
  x_{(p,\dot{a},q)}^i.
\end{multline*}
The size of $\Psi_\mathcal{B}(\vec{\alpha}, \vec{\sigma})$ is
dominated by the size of $\varphi_{\text{flows}}(\vec{\sigma}, p,
\vec{f}, \vec{s}, \vec{t})$ which is $O(k^2|\mathcal{B}|)$. \qed
\end{proof}
Note that it is easy to modify $\Psi_\mathcal{B}$ in order to have
$q_0$ as a free variable. By combining $\Psi_\mathcal{B}$ with
Lemma~\ref{lem:parikh-to-effect}, we obtain the following corollary.
%
\begin{corollary}
  \label{cor:reachability-presburger}
  Let $\mathcal{A}$ be a \ZVASSr \ and $p,q\in Q$. There exists a
  logarithmic-space computable existential Presburger
  formula\footnote{Here, we allow $\vec{v}$ and $\vec{w}$ to be
    interpreted over $\zint$, which can easily be achieved by
    representing an integer as the difference of two natural numbers.}
  $\Phi_{\mathcal{A}}(p,q,\vec{v},\vec{w},\vec{\alpha},\vec{\sigma})$
  such that $(p,q,\vec{v},\vec{w},\vec{\alpha},\vec{\sigma})\in
  \eval{\Phi_{\mathcal{A}}}$ if, and only if, there is $\gamma\in
  (\Sigma\dotcup R)^*$ such that
  %
    $\tilde{p}(\vec{v}) \stackrel{\gamma}{\rightarrow}_\mathcal{A}
    \tilde{q}(\vec{w})$
  and $(\vec{\alpha}, \sigma)\in \Pi(\gamma)$, where $\sigma(i)=\vec{\sigma}(i)$.
\end{corollary}
In particular, this implies that the reachability set of \ZVASSr \ is
semi-linear, and that reachability in \ZVASSr \ is \NP-complete.

\section{Inclusion for \ZVASS}
\label{sec:inclusion}

In this section, we show the following theorem.

\begin{theorem}
  Inclusion for \ZVAS \ is $\NP$-hard and in $\mathsf{\Pi}_2^\P$,
  and \coNEXP-complete for \ZVASS \ and \ZVASSr.
\end{theorem}

The upper bounds follow immediately from the literature. For \ZVAS
\ we observe that we are asking for inclusion between linear
sets. Huynh~\cite{Huy86} shows that inclusion for semi-linear sets is
$\mathsf{\Pi}_2^\P$-complete, which yields the desired upper
bound. Regarding inclusion for \ZVASSr, from
Corollary~\ref{cor:reachability-presburger} we have that the
reachability set of a \ZVASSr \ is $\Sigma_1$-PA definable. Let
$\mathcal{A},\mathcal{B}$ be \ZVASSr \ in dimension $d$,
$q(\vec{v})\in C(\mathcal{A})$, $p(\vec{w})\in C(\mathcal{B})$, and
let $\phi_{\mathcal{A},q(\vec{v})}(\vec{x})$ and
$\phi_{\mathcal{B},p(\vec{w})}(\vec{x})$ be appropriate $\Sigma_1$-PA
formulas from Corollary~\ref{cor:reachability-presburger} with
$\vec{x}=(x_1,\ldots,x_d)$. We have
\begin{align*}
\mathit{reach}(\mathcal{A},q(\vec{v}))\subseteq
\mathit{reach}(\mathcal{B},p(\vec{w})) \Leftrightarrow \neg( \exists
\vec{x}. \phi_{\mathcal{A},q(\vec{v})}(\vec{x}) \wedge \neg
(\phi_{\mathcal{B},p(\vec{w})}(\vec{x}))) \text{ is valid.}
\end{align*}
Bringing the above formula into prenex normal form yields a $\Pi_2$-PA
sentence for which validity can be decided in $\coNEXP$~\cite{Haa14}.
For that reason we focus on the lower bounds in the remainder of this
section.

For \ZVAS, an \NP-lower bound follows straight-forwardly via a
reduction from the feasibility problem of a system of linear
Diophantine equations $A\vec{x}=\vec{b}, \vec{x}\ge \vec{0}$. Despite
some serious efforts, we could not establish a stronger lower bound.
Even though it is known that inclusion for semi-linear sets is
$\mathsf{\Pi}_2^\P$-hard~\cite{Huy85}, this lower bound does not seem
to carry over to inclusion for \ZVAS. 


\begin{lemma}
  Inclusion for \ZVASS \ is \coNEXP-hard even when numbers are encoded
  in unary.
\end{lemma}
\begin{proof}
We reduce from validity in $\Pi_2$-PA, which is $\coNEXP$-hard already
when numbers are encoded in unary~\cite{Grae89,Haa14}. To this end,
let $\Phi=\forall \vec{x}.\exists \vec{y}.\varphi(\vec{x},\vec{y})$ be
a formula in this fragment such that $\vec{x}$ and $\vec{y}$ are $m$-
and $n$-tuples of first-order variables, respectively. As discussed in
the introduction, with no loss of generality we may assume that
$\varphi(\vec{x},\vec{y})$ is a positive Boolean combination of $k$
terms $t_1,\ldots, t_k$ of the form $t_i= \vec{a}_i\cdot \vec{x} + z_i
\ge \vec{b}_i\cdot \vec{y}$ with $\vec{a}_i\in \zint^{m},\vec{b}_i\in
\zint^n$ and $z_i\in \zint$. In our reduction, we show how to
construct in logarithmic space \ZVASS \ $\mathcal{A},\mathcal{B}$ with
designated control states $q,p$ such that $\Phi$ is valid iff
$\mathit{reach}(\mathcal{A},q(\vec{0}))\subseteq
\mathit{reach}(\mathcal{B},
p(\vec{0}))$. Figure~\ref{fig:inclusion-hardness} illustrates the
structure of the \ZVASS \ $\mathcal{A}$ and $\mathcal{B}$.  A key
point behind our reduction is that the counters of $\mathcal{A}$ and
$\mathcal{B}$ are used to represent evaluations of left-hand and
right-hand-sides of the \emph{terms} of $\varphi(\vec{x},\vec{y})$.
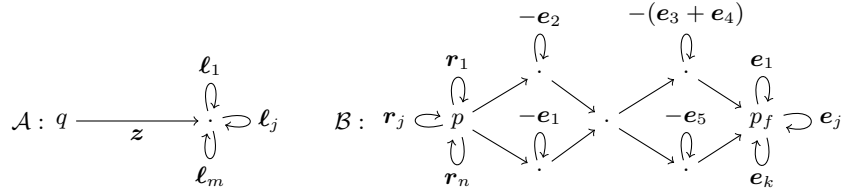
\begin{figure}[t]
\centering
\scalebox{1.0}{
  \begin{tikzpicture}[->,auto]

\node (q) {$q$};
\node [right=5em of q] (qA) {.};
\node [right=-3em of q] (labelA) {$\mathcal{A}:$};

\path (q) edge [swap] node {$\vec{z}$} (qA);
\path (qA) edge [loop above] node {$\vec{\ell}_1$} (qA);
\path (qA) edge [loop right] node {$\vec{\ell}_j$} (qA);
\path (qA) edge [loop below] node {$\vec{\ell}_m$} (qA);

\node [right=15em of q] (qp) {$p$};
\node [right=-6em of qp] (labelB) {$\mathcal{B}:$};
\node [above=1em of qp] (fq) {};
\node [below=1em of qp] (fqp) {};
\node [right=2.5em of fq] (q1) {.};
\node [right=2.5em of fqp] (qp1) {.};
\node [right=5em of qp] (q2) {.};
\node [right=5em of q1] (q3) {.};
\node [right=5em of qp1] (qp3) {.};
\node [right=5em of q2] (q4) {$p_f$};

\path (qp) edge [swap] node {} (q1);
\path (qp) edge [swap] node {} (qp1);
\path (q1) edge [swap] node {} (q2);
\path (qp1) edge [swap] node {} (q2);
\path (q2) edge [swap] node {} (q3);
\path (q2) edge [swap] node {} (qp3);
\path (q3) edge [swap] node {} (q4);
\path (qp3) edge [swap] node {} (q4);

\path (qp) edge [loop above] node {$\vec{r}_1$} (qp);
\path (qp) edge [loop left] node {$\vec{r}_j$} (qp);
\path (qp) edge [loop below] node {$\vec{r}_n$} (qp);

\path (q1) edge [loop above] node {$-\vec{e}_2$} (q1);
\path (qp1) edge [loop above] node {$-\vec{e}_1$} (qp1);

\path (q3) edge [loop above] node {$-(\vec{e}_3 + \vec{e}_4)$} (q1);
\path (qp3) edge [loop above] node {$-\vec{e}_5$} (qp1);

\path (q4) edge [loop above] node {$\vec{e}_1$} (qp);
\path (q4) edge [loop right] node {$\vec{e}_j$} (qp);
\path (q4) edge [loop below] node {$\vec{e}_k$} (qp);


  \end{tikzpicture}
}
  \caption{Illustration of the approach to reduce validity of a
    $\Pi_2$-PA formula $\Phi=\forall \vec{x}.\exists \vec{y}.(t_1 \vee
    t_2) \wedge ((t_3 \wedge t_4) \vee v_5)$ to inclusion for \ZVASS.}
  \label{fig:inclusion-hardness}
\end{figure}


In Figure~\ref{fig:inclusion-hardness}, we have that $\vec{z}\in
\zint^k$ is such that $\vec{z}(i)=z_i$. For $j\in [m]$,
$\vec{\ell}_j\in \zint^{k}$ is such that
$\vec{\ell}_j(i)=a_{i}(j)$. Likewise, for $j\in [n]$, $\vec{r}_j\in
\zint^{k}$ is such that $\vec{r}_j(i)=b_{i}(j)$. When moving away from
state $q$, $\mathcal{A}$ adds the absolute term of each $t_i$ to the
respective counters. It can then choose any valuation of the $\vec{x}$
and thus stores the corresponding values of the left-hand sides of
each $t_i$ in the counters. Now $\mathcal{B}$ has to match the choice
of $\mathcal{A}$. To this end, it can first loop in state $p$ in order
to guess a valuation of the $\vec{y}$ and update the values of the
counters accordingly, which now correspond to the right-hand sides of
the $t_i$. Along a path from $p$ to $p_f$, $\mathcal{B}$ may, if
necessary, simulate the Boolean structure of $\varphi$: conjunction is
simulated by sequential composition and disjunction by branching. For
every conjunct of $\varphi$, $\mathcal{B}$ can non-deterministically
decrement all but one term of every disjunct. Finally, once
$\mathcal{B}$ reaches $p_f$, it may non-deterministically increase the
value corresponding to the right-hand sides of every term in order to
precisely match any value reached by $\mathcal{A}$. From this example,
it is now clear how to construct $\mathcal{A}$ and $\mathcal{B}$ from
$\Phi$ in general in logarithmic space such that $\Phi$ is valid if,
and only if, $\mathcal{B}$ has a run beginning in $p(\vec{0})$ that
matches the counter values reached by any run of $\mathcal{A}$
beginning in $q(\vec{0})$. Obviously, the the converse direction holds
as well.\qed
\end{proof}



\section{Concluding Remarks}
We studied reachability, coverability and inclusion problems for
various classes of \ZVASS, \emph{i.e.}, \VASS \ whose counter values
range over $\zint$. Unsurprisingly, the complexity of those decision
problems is lower for \ZVASS \ when compared to \VASS. However, the
extend to which the complexity drops reveals an element of surprise:
coverability and reachability for \VASS \ in the presence of resets
are $\mathbf{F}_\omega$-complete and undecidable, respectively, but
both problems are only \NP-complete for \ZVASSr.
For the upper bound, we provided a generalization of Parikh images
which we believe is a technical construction of independent interest.

Throughout this paper, the dimension of the \ZVASS \ has been part of
the input. A natural line of future research could be to investigate
the complexity of the problems we considered in fixed dimensions.

\subsubsection*{Acknowledgments.} We would like to thank the anonymous
referees, Sylvain Schmitz and Philippe Schnoebelen for their helpful
comments and suggestions on an earlier version of this paper.


\bibliographystyle{plain} 
\bibliography{bibliography}


\newpage
\begin{appendix}
  \section{Missing Proofs from Section~\ref{sec:def}}

\subsection{Proof of Lemma~\ref{lem:np-hardness}}

\begin{lemma}
  \label{lem:zvas-reachability-hardness}
  Reachability in \ZVAS \ is \NP-hard already when numbers are encoded
  in unary.
\end{lemma}
\begin{proof}
  Let $S:\exists \vec{x}.A\vec{x} = \vec{b}, \vec{x}\ge \vec{0}$ be a
  system of linear Diophantine equations such that $A$ consists of row
  vectors $\vec{a}_1,\ldots,\vec{a}_n$. Determining whether $S$ is
  valid is a well-known $\NP$-hard problem even when numbers are
  encoded in unary~\cite{GJ79}. From $S$, we can easily construct a
  \ZVASS \ $\mathcal{A}$ with one control state $q$ such that
  $q(\vec{0}) \rightarrow^*_\mathcal{A} q(\vec{b})$ if, and only if,
  $S$ is valid as follows: for every $\vec{a}_i$, $\mathcal{A}$ has a
  self-loop reading the alphabet symbol $a_i$ and adding $\vec{a}_i$ to
  the counter. Given a word $\gamma$ witnessing $q(\vec{0})
  \rightarrow^*_\mathcal{A} q(\vec{b})$, counting the numbers of times
  each $a_i$ occurs along $\gamma$ yields a valuation of $\vec{x}$
  such that $A\vec{x}=\vec{b}$. Conversely, any valuation of $\vec{x}$
  such that $A\vec{x}=\vec{b}$ gives rise to a run $q(\vec{0})
  \rightarrow^*_\mathcal{A} q(\vec{b})$.\qed
\end{proof}

\begin{remark}
  If $\mathcal{A}$ is constructed as above, we have that $S$ is valid
  if, and only if, $\{ \lambda\vec{b} : \lambda\in \nat \} \subseteq
  \mathit{reach}(\mathcal{A},q(\vec{0}))$. This shows that inclusion
  for \ZVAS \ is \NP-hard. 
\end{remark}

\subsection{Proof of Lemma~\ref{lem:undecidability}}
\label{undec}

\begin{lemma}
  Let $\mathfrak{D}_d$ be the set of all diagonal matrices in
  dimension $d$. Coverability in \ZRM$(\mathfrak{D}_d)$ is undecidable
  already for $d=4$.
\end{lemma}

\begin{proof}
  Reachability for \ZRM$(\mathfrak{D}_2)$ is undecidable as announced
  in~\cite{FGH13}. This result has been obtained by J.\ Reichert and
  has not yet appeared in written format. For the sake of
  completeness, here we first repeat Reichert's argument.

  Undecidability is shown via reduction from the undecidable Post
  Correspondence Problem (PCP). Given $u_1, \dots, u_n, v_1, \dots,
  v_n \in \{0,1\}^*$, PCP asks whether there are some $i_1, \dots,
  i_p$ ($p>0$) such that $u_{i_1}\cdots u_{i_p} = v_{i_1}\cdots
  v_{i_p}$. Below, we define a \ZRM($\mathfrak{D}_2$) $\mathcal{A} =
  (\{ q_0, q_f\} \cup Q, \{ 0, 1, \tilde{0}, \tilde{1}, \#\}, 2,
  \Delta, \tau)$ such that there is a run from $q_0(\vec{0})$ to
  $q_f(\vec{0})$ in $\mathcal{A}$ if, and only if, there is a solution
  to the above PCP instance:
  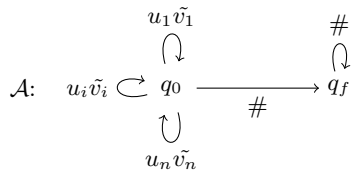
\begin{figure}[htbp]
\centering
\scalebox{1.0}{
  \begin{tikzpicture}[->,auto]

\node[circle] (q0) {$q_0$};
\node [right=5em of q0] (qf) {$q_f$};
\node [right=-8em of q0] (label) {$\mathcal{A}$:};

\path (q0) edge [swap] node {$\#$} (qf);
\path (q0) edge [loop above] node {$u_1\tilde{v_1}$} (q0);
\path (q0) edge [loop left] node {$u_i\tilde{v_i}$} (q0);
\path (q0) edge [loop below] node {$u_n\tilde{v_n}$} (q0);
\path (qf) edge [loop above] node {$\#$} (qf);


  \end{tikzpicture}
}
\vspace*{-1em}\caption{The \ZRM($\mathfrak{D}_2$) $\mathcal{A}$ used
  for the reduction from PCP.}
\label{fig-problems}
\vspace*{-1em}
\end{figure}

  $\mathcal{A}$ has $n$ self-loops on $q_0$, and each of these loops
  is labeled by a word $w=u_i\tilde{v_i}$. This, of course, actually
  corresponds to a path with $|w|$ states such that the path reads
  $w$. We now define $\tau$ as:
  \begin{align*}
  \tau(0)(\vec{v}) & = 
    \begin{pmatrix}
      2 \ 0 \\
      0 \ 1 
    \end{pmatrix} \vec{v} & \tau(\tilde{0})(\vec{v}) & = 
    \begin{pmatrix}
      1 \ 0 \\
      0 \ 2 
    \end{pmatrix} \vec{v}\\
    \tau(1)(\vec{v}) & = 
    \begin{pmatrix}
      2 \ 0 \\
      0 \ 1 
    \end{pmatrix} \vec{v} + 
    \begin{pmatrix}
      1\\
      0 
    \end{pmatrix}
    & \tau(\tilde{1})(\vec{v}) & = 
    \begin{pmatrix}
      1 \ 0 \\
      0 \ 2 
    \end{pmatrix} \vec{v} + 
    \begin{pmatrix}
      0\\
      1 
    \end{pmatrix}\\
    \tau(\#)(\vec{v}) & = \vec{v} - 
    \begin{pmatrix}
      1 \\
      1
    \end{pmatrix}
  \end{align*}
  
  The idea is that to reach $q_f(\vec{0})$ from $q_0(\vec{0})$, the
  two counters must be equal when leaving state $q_0$. Thinking of
  counter values encoded in binary, the counters represent the
  concatenation of the $u_i$ and the $v_i$, respectively, since in
  binary, multiplying by $2$ corresponds to concatenating $0$, and
  multiplying by $2$ and adding $1$ corresponds to concatenating
  $1$. Looping non-deterministically on $q_0$, the machine ``guesses''
  an order to make the two counters, \emph{i.e.}, words, equal.

  Note that the only matrices appearing in $\mathcal{A}$ are diagonal,
  and of dimension 2. By application of Lemma~\ref{equiv}, we obtain
  that coverability is henceforth undecidable for matrices from
  $\mathfrak{D}_4$. \qed
\end{proof}

Note that reachability for $\mathfrak{D}_1$ is shown decidable in,
which implies that coverability is decidable in this setting as
well. However, coverability for $\mathfrak{D}_2$ and $\mathfrak{D}_3$
remain an open problems. It is surprising to have such a complexity
gap between \ZRM \ with diagonal matrices and \ZRM \ with diagonal
matrices with only zeros and ones, which respectively make our
problems undecidable and \NP-complete. Thus, it is natural to wonder
whether some decidable class of matrices lies in between.




  \section{Missing Proofs from Section~\ref{sec:zvass}}

\subsection{Proof of Lemma \ref{lem:parikh-to-effect}}

\begin{lemma}
  Let $\mathcal{A}$ be a \ZVASSr, $\vec{v}\in \zint^d$, $\gamma \in
  (\Sigma\dotcup R)^*$, $(\vec{\alpha}_0,\vec{\alpha}_1, \dots,
  \vec{\alpha}_d, \sigma)\in\Pi(\gamma)$ a generalized Parikh image of $\gamma$ and
  $B\in \zint^{d\times n}$ the matrix whose columns are the vectors
  $\vec{b}_i$. Then the following holds:
  \begin{align*}
    \tau(\gamma)(\vec{0}) =  \sum_{1\le i\le d} 
    (B\vec{\alpha}_{i-1})_{|\{\sigma(i), \dots, \sigma(d)\}} + B\vec{\alpha}_d.
  \end{align*}
\end{lemma}

\begin{proof}
Let $p$ be the number introduced in Definition \ref{def}, we prove the following stronger statement by induction on $j \in [p, d]$:
\begin{align*}
\tau(\gamma_p r_{\sigma(p+1)} \gamma_{p+1} \dots r_{\sigma(j)} \gamma_j)(\vec{0})_{|\{\sigma(j+1), \dots, \sigma(d)\}} = \sum_{i = 0}^{j} (B\vec{\alpha}_i)_{|\{\sigma(i+1), \dots, \sigma(d)\}}
\end{align*}
where $\gamma = \gamma_p r_{\sigma(p+1)} \gamma_{p+1} \dots r_{\sigma(d)} \gamma_k$ is the decomposition introduced in Definition \ref{def}. We then conclude by taking $j = d$. 
\begin{itemize}
\item Base case $j=p$:

Let $S= \{ \sigma(p+1), \dots, \sigma(d) \}$. Since only resets $r_i$ for $i \in S$ occurs in $\gamma_p$ by definition of the decomposition, and since addition is commutative and associative, only the number of times each letter appear is important. Therefore:
\begin{align*}
\tau(\gamma_p)(\vec{0})_{|S} = \sum_{i = 1}^n |\gamma_p|_{a_i} . (\vec{b}_i)_{|S} = (B \vec{\alpha}_p)_{|S} = \sum_{i=0}^p (B\vec{\alpha}_i)_{|\{\sigma(i+1), \dots, \sigma(d)\}}
\end{align*}
Remember that for $i < p$, $\vec{\alpha}_i = \vec{0}$, which explains the last equality.
\item Induction step: Let $S = \{ \sigma(j+1), \dots, \sigma(d) \}$
  and $S' = \sigma(j) \cup S$ and $\gamma' = \gamma_p r_{\sigma(p+1)}
  \dots \gamma_{j-1}$:
\begin{align}
 \tau(\gamma_p & r_{\sigma(p+1)} \dots r_{\sigma(j)} \gamma_j)(\vec{0})_{|S}\notag  \\
    & = \tau(\gamma'r_{\sigma(j)}\gamma_j)(\vec{0})\notag \\
    & = \tau(r_{\sigma(j)}\gamma_j)(\tau(\gamma')(\vec{0}))_{|S} \\
    & = \tau(\gamma_j)( [\tau(\gamma')(\vec{0})]_{|\sigma(j)})(\vec{0}))_{|S} \\
    & = [(\tau(\gamma')(\vec{0})))_{|\sigma(j)} + \tau(\gamma_j)(\vec{0})]_{|S}  \\
    & = [(\tau(\gamma')(\vec{0}))_{|\sigma(j)}]_{|S} + \tau(\gamma_j)(\vec{0})_{|S} \notag \\
    & = \tau(\gamma')(\vec{0})_{|S'} + \tau(\gamma_j)(\vec{0})_{|S} \notag \\
    & = \sum_{i = 0}^{j-1} (B\vec{\alpha}_i)_{|\{\sigma(i+1), \dots, \sigma(d)\}} + (B\vec{\alpha}_j)_{|S} \\
    & = \sum_{i = 0}^{j} (B\vec{\alpha}_i)_{|\{\sigma(i+1), \dots, \sigma(d)\}}, \notag
\end{align}
where
\begin{enumerate}[(1)]
\item by definition of $\tau$
\item by definition of $\tau(r_{\sigma(j)})$
\item as $\gamma_j$ has only resets from $S$
\item by induction hypothesis.
\end{enumerate}
\end{itemize}
\qed
\end{proof}

\subsection{Proof of Corollary \ref{cor:reachability-presburger}}

Throughout this section, let $\mathcal{A}=(Q,\Sigma\dotcup
R,d,\Delta,\tau)$ be a \ZVASSr. Before we give the proof of
Corollary~\ref{cor:reachability-presburger}, we prove the following
lemma.
\begin{lemma}
  \label{lem:counters-logic}
  There exists a logarithmic-space computable existential Presburger
  formula $\varphi_\text{counters}(\vec{\alpha}, \vec{\sigma}, p,
  \vec{v}, \vec{v}')$ such that $(\vec{\alpha}, \vec{\sigma}, p,
  \vec{v}, \vec{v}')\in \eval{\varphi_\text{counters}}$ if, and only
  if, there is a word $\gamma\in (\Sigma\dotcup R)^*$ such that
  $\tau(\gamma)(\vec{v}) = \vec{v}'$ and $(\vec{\alpha}, \vec{\sigma})
  \in \Pi(\gamma)$ with $p$ being the number introduced in
  Definition~\ref{def}.
\end{lemma}

\begin{proof}
  In Presburger arithmetic, the equality
  $\tau(\gamma)(\vec{v})=\vec{v}'$ is actually represented by $d$
  equalities $\tau(\gamma)(\vec{v})(i) = \vec{v}'(i)$ for $i\in
  [d]$. Lemma~\ref{lem:parikh-to-effect} states that for any $i \in
  [d]$, $\tau(\gamma)(\vec{v})(i) = \sum_{j = 0}^d
  (B\vec{\alpha}_j)_{|\{\sigma(j+1), \dots, \sigma(d)}(i) =
  \sum_{j=0}^d \lambda_{ij}(B\vec{\alpha}_j)(i)$ where
  \begin{align*}
    \lambda_{ij} & =
    \left\{ \begin{array}{ll} 0 &\text{ iff } i\in \{ \sigma(j+1),
      \dots, \sigma(d) \} \text{ iff } \sigma^{-1}(i) \ge (j+1) \\ 1
      &\text{ otherwise}.
    \end{array} \right.
  \end{align*}
  This last equality is not a syntactically correct Presburger term
  since it is quadratic instead of linear. We therefore introduce
  intermediate variables $\beta_j^i$ to compute the partial sums
  $\beta_j^i = \sum_{k=0}^j \lambda_{ik} (B\vec{\alpha}_k)(i)$:
  \begin{multline*}
    \varphi_\text{counters} (\vec{\alpha}, \vec{\sigma}, p, \vec{v}, \vec{v}') = 
    \exists \vec{\beta}. \exists \vec{\nu}. 
    \bigwedge_{i=1}^d \beta_0^i = 0 \wedge \vec{v}'(i) = \beta_d^i + \vec{\nu}(i) 
    \wedge
    \\\wedge \bigwedge_{k=1}^d (\sigma(k) = i) \rightarrow \bigwedge_{j=1}^d 
    (k > j \rightarrow \beta_j^i = \beta_{j-1}^i)  \wedge 
      (k \le j \rightarrow \beta_j^i = \beta_{j-1}^i + (B\vec{\alpha}_j)(i)
    \wedge \\
    \wedge (k > p \rightarrow \vec{\nu}(i) = 0) \wedge
    (k \le p \rightarrow \vec{\nu}(i) = \vec{v}(i)).
  \end{multline*}
  When reading this formula, one should see $k$ as $\sigma^{-1}(i)$,
  and therefore recognize the second line to be the condition
  expressed above. As $\beta_j^i$ are the partial sums up to $j$,
  $\beta_d^i$ represents the complete sum on the dimension $i$. The
  final vector $\vec{v}'$ is thus equal to the vector made of the
  $\beta_d^i$ plus the starting vector $\vec{v}$ in which the right
  components have been erased: this is the vector $\vec{\nu}$.  \qed
\end{proof}

\begin{corollary}[Corollary~\ref{cor:reachability-presburger} in the main text]
  Let $\mathcal{A}$ be a \ZVASSr \ and $p,q\in Q$. There exists a
  logarithmic-space computable existential Presburger formula
  $\Phi_{\mathcal{A}}(q',q,\vec{v},\vec{w},\vec{\alpha},\vec{\sigma})$
  such that $(p,q,\vec{v},\vec{w},\vec{\alpha},\vec{\sigma})\in
  \eval{\Phi_{\mathcal{A}}}$ if, and only if, there is $\gamma\in
  (\Sigma\dotcup R)^*$ such that
    $\tilde{q'}(\vec{v}) \stackrel{\gamma}{\rightarrow}_\mathcal{A}
    \tilde{q}(\vec{w})$
  and $(\vec{\alpha}, \sigma)\in \Pi(\gamma)$, where $\sigma(i)=\vec{\sigma}(i)$.
\end{corollary}

\begin{proof}
  Note $\Psi'_\mathcal{B}$ for the formula $\Psi_\mathcal{B}$
  (\emph{cf.} Theorem~\ref{gpi}) without the quantification on $p$,
  and with additionnal free variables q and q' for the initial and
  final states.  By Lemma~\ref{lem:counters-logic}, we conclude with:
\begin{align*} 
  \Phi_{\mathcal{A}}(q',q,\vec{v},\vec{w},\vec{\alpha},\vec{\sigma}) = 
  \exists p. \ \Psi'_{\mathcal{B}}(\vec{\alpha}, \vec{\sigma}, p, q, q') \wedge 
  \varphi_\text{counters}(\vec{\alpha}, \vec{\sigma}, p, \vec{v}, \vec{v}').
\end{align*}
\qed
\end{proof}

\end{appendix}

\end{document}